\documentclass[natbib,sort,smallextended,envcountsame,envcountsect,numbook]{svjour3}
\pdfoutput=1
\smartqed
\usepackage{amssymb,amsmath}
\usepackage[utf8]{inputenc}
\usepackage[T1]{fontenc}
\usepackage[english]{babel}
\usepackage{txfonts}
\usepackage{microtype}
\usepackage{mathtools}
\usepackage{mathtools}
\usepackage{hyperref}
\usepackage[capitalize,noabbrev]{cleveref}
\usepackage{doi}

\Crefname{observation}{Observation}{Observations}
\Crefname{assumption}{Assumption}{Assumptions}
\Crefname{problem}{Problem}{Problems}

\usepackage{xcolor}
\usepackage{booktabs}
\usepackage{paralist}
\usepackage[textsize=scriptsize]{todonotes}
\hypersetup{colorlinks,allcolors=blue}
\spnewtheorem{observation}[theorem]{Observation}{\bfseries}{\normalfont}
\spnewtheorem{assumption}[theorem]{Assumption}{\bfseries}{\normalfont}
\spnewtheorem{algorithm}[theorem]{Algorithm}{\bfseries}{\normalfont}
\newcommand{\HC}{H}
\newcommand{\bigO}{O}

\newcommand{\ps}{pre-stay}

\newcommand{\las}{length assignment}
\newcommand{\dis}{displacement}

\newcommand{\depo}{v^*}
\newcommand{\ROS}{\textsc{ROS}}
\newcommand{\ROSUPT}{\textsc{ROS-UET}}

\newcommand{\sched}{S}
\newcommand{\loc}{\text{loc}}

\newcommand{\dist}{\text{dist}}
\newcommand{\Js}{\mathcal J}
\newcommand{\Ms}{\mathcal M}

\newcommand{\njobs}{n}
\newcommand{\nmach}{m}
\newcommand{\nverts}{g}

\newcommand{\jobsin}[1]{n_{#1}}

\newcommand{\arr}{a}
\newcommand{\dep}{b}
\newcommand{\stops}{s}

\newcommand{\Pro}{P}
\newcommand{\Gr}{G}
\def\pvstays{\ensuremath{{m+2g-3}}}
\def\totstays{\ensuremath{{m^2+(2g-3)m}}}

\title{A fixed-parameter algorithm for a routing open shop problem: unit processing times, few machines and locations\thanks{A preliminary version of this work appeared at CSR'16 \citep{BP16}.  This version provides simpler proofs, a faster algorithm that runs in \(O(n\log n)\)~time for fixed \(m\) and~\(g\), and replaces the incorrect upper bound given in Lemma~5.5 of the old version by an asymptotically stronger upper bound.}}
\titlerunning{A fixed-parameter algorithm for routing open shop with unit processing times}

\author{René van Bevern\thanks{René van Bevern is supported
by the Russian Foundation for Basic Research, grant~16-31-60007 mol\textunderscore{}a\textunderscore{}dk.}
\and
Artem V.\ Pyatkin\thanks{Artem V.\ Pyatkin is supported
by the Russian Foundation for Basic Research, grants~17-01-00170 and~15-01-00976.}}

\institute{René van Bevern \and Artem V.\ Pyatkin \at Novosibirsk
  State University, Novosibirsk, Russian Federation,
  \email{rvb@nsu.ru, artem@math.nsc.ru} \and René van Bevern \and
  Artem V.\ Pyatkin \at Sobolev Institute of Mathematics of the
  Siberian Branch of the Russian Academy of Sciences, Novosibirsk,
  Russian Federation}

\date{Submitted: \today{}}

\begin{document}

\maketitle
\begin{abstract}
  \noindent\looseness=-1
  The open shop problem is
  to find a minimum makespan schedule
  to process each job~\(J_i\)
  on each machine~\(M_q\)
  for $p_{iq}$~time
  such that,
  at any time,
  each machine processes at most one job
  and
  each job is processed by at most one machine.
  We study a problem variant in which the jobs are located
  in the vertices
  of an edge-weighted graph.
  The weights determine the time needed
  for the machines
  to travel between jobs in different vertices.
  We show that the problem with \(m\)~machines and
  \(n\)~unit-time jobs in \(g\)~vertices
  is solvable in \(2^{\smash{\bigO(\nverts\nmach^2\log\nverts\nmach)}}+O(mn\log n)\)~time.
  \keywords{%
    routing\and
    scheduling\and
    setup times\and
    UET\and
    NP-hard problem
  }
\end{abstract}
\section{Introduction} 
\noindent
\citet{GS76} introduced the open shop problem:
given
a set~$\Js:=\{J_1,\dots,J_n\}$ of jobs,
a set~$\Ms:=\{M_1,\dots,M_m\}$ of machines,
and the processing time~$p_{iq}$ that job~$J_i$ needs on machine~$M_q$,
the goal is
to process each job on each machine
in a minimum total amount of time
such that
each machine processes at most one job at a time
and
each job is processed by at most one machine at a time.

\citet{ABC06} introduced a problem variant
where the jobs are located
in the vertices of an edge-weighted graph.
The edge weights determine
the time needed for the machines to travel between jobs in different vertices.
In their setting, all machines have the same travel speed
and
the travel times are symmetric.
Initially,
the machines are located in a vertex called the \emph{depot}.
The goal is to minimize
the time needed for processing
each job by each machine and
returning all machines to the depot.
This problem variant models, for example, scenarios
where machines or specialists
have to perform maintenance work
on objects in several places.
The travel times have also been interpreted
as sequence-dependent machine-independent
batch setup times \citep{ANCK08,ZW06}.

In order to formally define the problem,
we first give a formal model for a transportation network
and machine routes.
Throughout this work,
we use~\(\mathbb N:=\{0,1,\dots\}\).

\begin{definition}[network, depot, routes]
  A \emph{network}~\(\Gr=(V,E,c,v^*)\) consists of
  a set~\(V\) of \emph{vertices} and
  a set~\(E\subseteq\{\{u,v\}\mid u\ne v, u\in V, v\in V\}\) of \emph{edges}
  such that \((V,E)\) is a connected simple graph,
  travel times~$c\colon E\to\mathbb N\setminus\{0\}$,
  and a vertex~\(v^*\in V\) called the \emph{depot}.
  We denote the number of vertices in a network
  by~\(\nverts\). 
  
  A \emph{route with $\stops$~stays} is
  a sequence~$R:=(R_k)_{k=1}^\stops$
  of \emph{stays}~$R_k=(\arr_k,v_k,\dep_k)\in\mathbb N\times V\times\mathbb N$
  from time~$\arr_k$ to time~$\dep_k$
  in vertex~$v_k$ for $k\in\{1,\dots,\stops\}$
  such that \(v_1=v_\stops=\depo\), \(a_1=0\),
  \begin{align*}
    a_k&\leq b_k&&\text{for all }k\in\{1,\dots,\stops\},\\
    \{v_k,v_{k+1}\}&\in E&&\text{for all }k\in\{1,\dots,\stops-1\}\text{, and}\\
    a_{k+1}&=\dep_k+c(v_k,v_{k+1})&&\text{for all }k\in\{1,\dots,\stops-1\}.
  \end{align*}
  The \emph{length} of~$R$ is the end~$\dep_\stops$ of the last stay.  
\end{definition}

\noindent
\looseness=-1
Note that a route is actually fully determined
by the \(v_k\) and \(b_k\) for each~\(k\in\{1,\dots,\stops\}\),
yet it will be convenient to refer to both
\emph{arrival time~\(a_k\)} and \emph{departure time~\(b_k\)} directly.
We now define the routing open shop variant (ROS) introduced by \citet{ABC06}.

\begin{definition}[\ROS{}]\label[definition]{def:ros}
  An \emph{instance}~\(I=(\Gr,\Js,\Ms,\allowbreak \loc,\Pro)\)
  of the \ROS{} problem
  consists of a network~$\Gr=(V,E,c,v^*)$,
  a set $\Js=\{J_1,\dots,J_n\}$ of~\emph{jobs},
  a set $\Ms=\{M_1,\dots, M_m\}$ of~\emph{machines},
  job \emph{locations}~$\loc\colon \Js\to V$,
  and an \((n\times m)\)-matrix~\(\Pro=(p_{iq})\)
  determining the \emph{processing time}~$p_{iq}\in\mathbb N$
  of each job~$J_i$ on each machine~$M_q$.

  \looseness=-1
  A \emph{schedule} $\sched\colon \Js{}\times\Ms{}\to\mathbb N$ is
  a function determining the \emph{start time}~$\sched(J_i, M_q)$
  of each job~$J_i$ on each machine~$M_q$.
  A job~$J_i$ is \emph{processed}
  by a machine~$M_q$
  in the half-open time interval~$[\sched(J_i,M_q),\sched(J_i,M_q)+p_{iq})$.
  A schedule is \emph{feasible} if and only if
  \begin{enumerate}[(i)]
  \item\label{def:feas2} no machine~$M_q$ processes two jobs~$J_i\ne J_j$ at the same time, that is, $\sched(J_i,M_q)+p_{iq}\leq \sched(J_j,M_q)$ or $\sched(J_j,M_q)+p_{jq}\leq \sched(J_i,M_q)$ for all jobs~$J_i\ne J_j$ and machines~$M_q$,
  \item\label{def:feas3} no job~$J_i$ is processed by two machines~$M_q\ne M_r$ at the same time, that is, $\sched(J_i,M_q)+p_{iq}\leq \sched(J_i,M_r)$ or $\sched(J_i,M_r)+p_{ir}\leq \sched(J_i,M_q)$ for all jobs~$J_i$ and machines~$M_q\ne M_r$, and
  \item\label{def:feas4} there are routes~\((R^q)_{M_q\in\Ms}\)
    \emph{compatible with~\(\sched\)}, that is,
    for each job~$J_i$ and each machine~$M_q$
    with route~$R^q=(R_k^q)_{k=1}^{s_q}$,
    there is a $k\in\{1,\dots,\stops\}$
    such that
    $R_k^q=(\arr_k^q,\loc(J_i),\dep_k^q)$
    with $\arr_k^q\leq \sched(J_i,M_q)$ and $\sched(J_i,M_q)+p_{iq}\leq\dep_k^q$.
  \end{enumerate}
  The \emph{makespan} of a feasible schedule~\(\sched\)
  is the minimum value~\(L\)
  such that
  there are routes~$(R_{M_q})_{M_q\in\Ms}$ compatible with~\(S\)
  and each having length at most~\(L\).
  An \emph{optimal solution} to \ROS{}
  is a feasible schedule with minimum makespan.
\end{definition}

\noindent
\looseness=-1
Note that a schedule for \ROS{} only
determines the start time of each job on each machine,
not the times and destinations of machine movements.
Yet the start times of each job on each machine
fully determine the order
in which each machine processes its jobs.
Thus,
one can easily construct compatible routes if they exist:
each machine simply takes the shortest path
from one job to the next
if they are located in distinct vertices.

\paragraph{Preemption and unit processing times} 
The open shop problem is NP-hard even in the special cases
of $\nmach=3$~machines \citep{GS76} or if all processing times are one or two \citep{KSS12}.
Naturally, these results transfer to \ROS{} with
with $\nverts=1$~vertex.
\ROS{} remains (weakly) NP-hard
even for $\nverts=\nmach=2$ \citep{ABC06}
and there are approximation algorithms
both for this special and the general case of \ROS{}
\citep{ABC05,YLWF11,CKS13,Kon15}.
However, the open shop problem is solvable in polynomial time if
\begin{enumerate}[(1)]
\item\label{case1}job preemption is allowed, or
\item\label{case2}all jobs~$J_i$ have unit processing time~$p_{iq}=1$  on all machines~$M_q$.
\end{enumerate}

\looseness=-1
\noindent
It is natural to ask how these results
transfer to \ROS{}.
Regarding~\eqref{case1},
\citet{PCh12} have shown that
\ROS{} with allowed preemption
is solvable in polynomial time
if $\nverts=\nmach=2$,
yet NP-hard for $\nverts=2$ and
an unbounded number~$\nmach$ of machines.
Regarding \eqref{case2}, our work studies
the following special case of \ROS{} with unit execution times (\ROSUPT):

\begin{problem}
  By \ROSUPT{}, we denote \ROS{} restricted to instances
  where each job~$J_i$ has unit processing time~$p_{iq}=1$ on each
  machine~$M_q$.
\end{problem}

\noindent
\ROSUPT{} models scenarios where machines or specialists
process batches of roughly equal-length jobs
in several locations and
movement between the locations
takes significantly longer
than processing each individual job in a batch.
\ROSUPT{} is NP-hard even for $\nmach=1$ machine
since it generalizes the metric travelling salesman problem.
It is not obvious whether it is solvable in polynomial time
even when both \(\nverts\) and~\(\nmach\) are fixed.
We show that, in this case,
\ROSUPT{} is solvable even in \(O(n\log n)\)~time, that is, \ROSUPT{} is
\emph{fixed-parameter tractable} parameterized by~\(g+m\).

\paragraph{Fixed-parameter algorithms}
Fixed-parameter algorithms are an approach towards
efficiently and optimally solving NP-hard problems:
the main idea is to accept the exponential running time
for finding optimal solutions,
yet to confine it
to some small problem parameter~$k$.
A problem with parameter~$k$
is called \emph{fixed-parameter tractable~(FPT)}
if there is an algorithm that solves any instance~$I$
in $f(k)\cdot|I|^{O(1)}$~time,
where $f$~is an arbitrary computable function.
The corresponding algorithm is called \emph{fixed-parameter algorithm}.  
For more detail,
we refer the reader to the recent text book by \citet{CFK+15}.

\looseness=-1
Note that a fixed-parameter algorithm
running in \(O(2^k\cdot|I|)\)~time
runs in polynomial time for~\(k\in O(\log |I|)\),
whereas an algorithm with running time~\(O(|I|^k)\)
runs in polynomial time only if \(k\)~is constant.
The latter algorithm is not a fixed-parameter algorithm.

Recently, the field of fixed-parameter algorithmics
has shown increased interest in scheduling and routing
\citep{CMY+17,BBB+16b,BCH+15,BMNW15,BNS16,BF95,FM03,HK06,
  HKS+15,MW15,BKS17,BNSW14,DMNW13,GWY17,GJW16,GJS17,GMY13,JMS17,KM14,
  SBNW11,SBNW12},
whereas fixed-parameter algorithms for problems
containing elements of both routing and scheduling
seem rare \citep{BHKK07}.

\paragraph{Input encoding}
Encoding a \ROS{} instance
requires \(\Omega(\njobs\cdot\nmach+g)\)~bits
in order to encode the processing time of each of \(n\)~jobs
on each of \(m\)~machines
and the travel time along each of at least \(g-1\)~edges.
We call this the \emph{standard encoding}.
In contrast, a \ROSUPT{} instance
can be encoded using
\(\bigO(\nverts^2\cdot\log c_{\max}+\nverts\cdot\log\njobs)\) bits
by encoding only the number of jobs in each vertex, where \(c_{\max}\) is the maximum travel time.  We call this the \emph{compact encoding}.

All running times in this article are stated
for computing a minimum makespan schedule,
whose encoding requires \(\Omega(\njobs\cdot\nmach)\)~bits
for the start time of each job on each machine.
Thus,
outputting the schedule is impossible
in time polynomial in the size of the compact encoding.
We therefore assume
to get the input instance in standard encoding,
like for the general \ROS{} problem.

However, we point out that the \emph{decision version} of \ROSUPT{}
is fixed-parameter tractable
parameterized by~\(\nverts+\nmach\)
even when assuming the compact encoding:
our algorithm can decide whether
there \emph{exists} a schedule of given makespan~\(L\)
in \(2^{\smash{\bigO(\nverts\nmach^2\log\nverts\nmach)}}\cdot|I|\)~time
if \(I\)~is a \ROSUPT{} instance given in compact encoding
(when replacing line~\ref{lin:complsched} of \cref{alg:outline}
by ``return $L$'', it will simply output the minimum makespan
instead of constructing the corresponding schedule).

\paragraph{Organization of this work}
In \cref{sec:preprop},
we apply some basic preprocessing
that allows us to
assume  that travel times
satisfy the triangle inequality.
In \cref{sec:ubo},
we prove upper and lower bounds
on the minimum makespan of schedules
and on the number and length of stays
of routes compatible with some optimal schedule.
In \cref{sec:fpt}, we present our fixed-parameter algortihm.

\section{Preprocessing for metric travel times}\label{sec:preprop}
\noindent
In this section,
we transform \ROS{} instances
into equivalent instances
with travel times
satisfying the triangle inequality.
This will allow us to assume that,
in an optimal schedule,
a machine only stays in a vertex
if it processes at least one job there:
otherwise, it could take a ``shortcut'',
bypassing the vertex.

\begin{lemma}\label[lemma]{lem:metric}
  Let \(I\)~be a \ROS{} instance
  and
  let \(I'\)~be obtained from~\(I\)
  by replacing the network~\(\Gr=(V,E,c,v^*)\) in~\(I\)
  by the network~\(\Gr'=(V,E',c',v^*)\)
  such that
  \((V,E')\) is a complete graph and
  \(c':\{v,w\}\mapsto\dist_c(v,w)\),
  where \(\dist_c(v,w)\)~is the length
  of a shortest path
  between~\(v\) and~\(w\) in~\(\Gr\) with respect to~\(c\).

  Then,  any feasible schedule for~\(I\)
  is a feasible schedule for \(I'\)
  with the same makespan and vice versa.
\end{lemma}

\begin{proof}
  Clearly, any feasible schedule~\(S\) of makespan~\(L\) for~\(I\)
  is a feasible schedule of makespan at most~\(L\) for~\(I'\). %
  We show that any feasible schedule~\(S'\) of makespan~\(L\) for~\(I'\)
  is a feasible schedule of makespan at most~\(L\) for~\(I\).
  This is because, for any route~\(R'\) compatible with~\(S'\)
  in~\(I'\),
  we get a route~\(R\) of the same length compatible with~\(S'\) in~\(I\):
  between each pair of consecutive
  stays~\((a_i,v_i,b_i)\) and~\((a_{i+1},v_{i+1},b_{i+1})\) on~\(R'\)
  and a path~\(P=(w_1=v_i,w_2,\dots,w_\ell=v_{i+1})\)
  of length~\(c'(v_i,v_{i+1})\)
  with respect to~\(c\)  in~\(\Gr\),
  add zero-length stays in the vertices
  \(w_2,w_3,\dots,w_{\ell-1}\)
  of~\(P\).  This yields a route~\(R\) for~\(I\)
  of the same length as~\(R'\) since the end of the last stay has not changed.
\qed\end{proof}

\noindent
The travel times~\(c'\)
in the network~\(\Gr'\) created in \Cref{lem:metric}
satisfy the triangle inequality.
Thus,
one can assume that, except for the depot, a machine
never visits a vertex of~\(\Gr'\) that has no jobs.
Since one can delete such vertices,
in the following,
we will make the
following simplifying assumption without loss of generality.

\begin{assumption}\label[assumption]{obs:delverts}
  Let \(I\)~be a \ROS{} instance
  on a network~\(\Gr=(V,E,c,v^*)\).
  Then,
  \begin{enumerate}[(i)]
  \item\label{ass1} the travel times~\(c\) satisfy the triangle inequality,
  \item\label{ass2} each vertex~\(v\in V\setminus\{v^*\}\) contains at least one job.
  \end{enumerate}
\end{assumption}

\section{Upper and lower bounds on makespan, number and lengths of stays}\label{sec:ubo}
\noindent
In this section,
we show lower and upper bounds on the makespan
of optimal \ROSUPT{} schedules,
as well as
on the number and the lengths of stays
of routes compatible with optimal schedules.
These will be exploited
in our fixed-parameter algorithm.

By \cref{obs:delverts}\eqref{ass1},
the travel times in the network~\(G\) of a \ROSUPT{} instance
satisfy the triangle inequality.
Thus, the minimum cost of a cycle
visiting each vertex of~\(\Gr\) \emph{at least once}
coincides with the minimum cost of a cycle
doing so \emph{exactly once} \citep{Ser78},
that is, with that of a minimum-cost \emph{Hamiltonian cycle}.

A trivial lower bound
on the makespan of optimal \ROSUPT{} schedules
is given by the fact that,
in view of \cref{obs:delverts}\eqref{ass2},
each machine has to visit
all vertices at least once
and has to process $\njobs$~jobs.
A trivial upper bound is given by
the fact that the machines can process the jobs sequentially.
We thus obtain the following:
\begin{observation}\label[observation]{obs:lobo}
  Let \(I\)~be a \ROSUPT{} instance on a network~$\Gr=(V,E,c,v^*)$
  with a minimum-cost Hamiltonian cycle~\(\HC\).
   Then, the makespan of an optimal schedule to~\(I\) lies in
   \(
    \{c(\HC)+\njobs,\dots,c(\HC)+\njobs+\nmach-1\}.
    \)

    Moreover, a schedule with makespan $c(\HC)+\njobs+\nmach-1$ for
    \ROSUPT{} is computable in
    \(\bigO(\njobs\cdot\nmach)\)~time if a Hamiltonian
    cycle~\(\HC\) for the network of the input
    instance is also given as input.
\end{observation}

\noindent
We can improve the upper bound
on the makespan
if \(c(H)<min\{\njobs,\nmach\}\):
\begin{proposition}\label[proposition]{thm:artem}
  A feasible schedule of length \(2c(\HC)+\max\{\njobs,\nmach\}\)
  for \ROSUPT{} is computable in \(\bigO(m^2+mn+\nverts)\)~time
  if a Hamiltonian
  cycle~\(\HC\) for the network of the input
  instance is also given as input.
\end{proposition}
\begin{proof}
  Without loss of generality, assume that \(n\geq m\).  Otherwise, we
  can simply add
  $m-n$~additional jobs to the depot and finally remove them from the
  constructed schedule.  We will construct a feasible schedule~\(S\)
  of length~\(2c(\HC)+n\).

  Let~\(\HC=(v_1,v_2,\ldots, v_{\nverts},v_1)\),
  where \(v_1=v^* \)~is the depot.
  Without loss of generality,
  let the jobs~\(J_1,\dots,J_n\) be ordered so that,
  for jobs~\(J_i,J_j\) with~\(i\leq j\), one has \(\loc(J_i)=v_k\)
  and \(\loc(J_j)=v_\ell\) with~\(k\leq\ell\).
  That is, the first jobs are in~\(v_1\),
  then follow jobs in~\(v_2\), and so on.  
  We will construct our schedule from the
  matrix~\(S'=(s'_{iq})_{1\leq i\leq n,1\leq q\leq m}\), where
  \begin{align}
s'_{iq}&:=(i-q)\bmod n=
\begin{cases}
  n-q+i&\text{ if \(i< q\),}\\
  i-q&\text{ otherwise.}
\end{cases}\label{eq:shift}
  \end{align}
  \looseness=-1
  Figuratively, each row of~\(S'\) is a cyclic right-shift of the previous row.
Call a cell $s'_{iq}$ \textit{red} if $i<q$ and \textit{green} otherwise. Note that, if $s'_{iq}$ and $s'_{jr}$ are of the same color and $i<j$ or $r<q$, then $s'_{iq}<s'_{jr}$. Moreover, the number in a red cell is larger than the number in any green cell in the same row or column: if \(s'_{iq}\)~is red and \(s'_{jq}\) is green, then from
 \begin{align*}
   n+i&>j\text{\quad follows}\\
   s'_{iq}=n-q+i&>j-q=s'_{jq}
 \end{align*}
and if \(s'_{iq}\)~is red and \(s'_{ir}\)~is green, then from
\begin{align*}
  n-q&>-r\text{\quad follows}\\
s'_{iq}=n-q+i&>i-r=s'_{ir}.
\end{align*}
Let $c_k=\sum_{i=2}^{k}c(v_{i-1},v_i)$ be the travel time from $v_1$ to $v_k$ along~\(\HC\).  Clearly, the sequence~$(c_k)_{k=1}^s$ is non-decreasing and~$c_{\nverts}\leq c(\HC)$.  Our schedule~\(S\) is now given by
\[
S(J_i,M_q)=s_{iq}:=s'_{iq}+
\begin{cases}
  c_{k}&\text{ if $\loc(J_i)=v_k$ and $s_{iq}$ is green},\\
  c(\HC)+c_k&\text{ if $\loc(J_i)=v_k$ and $s_{iq}$ is red}.
\end{cases}
\]
Let us prove that this schedule is feasible in terms of \cref{def:ros}. Indeed, by construction, for two elements $s_{iq}$ and $s_{jr}$ with \(i=j\) or \(q=r\) and $s'_{iq}>s'_{jr}$, one has $s_{iq}>s_{jr}$ since the value added to $s'_{iq}$ is not smaller than the value added to $s'_{jr}$ due to our sorting of jobs by non-decreasing vertex indices and because the value added to any red cell is larger than any value added to a green cell.  Therefore, conditions~\eqref{def:feas2} and \eqref{def:feas3} are satisfied.

It remains to verify \eqref{def:feas4}, that is,
that there are compatible routes~\(R^q=((a_k^q,v_k^q,b_k^q))_{k=1}^{s_q}\)
for each machine~\(M_q\).
We let machine~\(M_q\) follow~\(\HC\) twice.
During the first stay~\((a_k^q,v_k^q,b_k^q)\) in a vertex~\(v_k^q\),
machine~\(M_q\) processes all jobs~\(J_i\)
such that \(s_{iq}\) is green.
During the second stay~\((a_{g+k}^q,v_k^q,b_{g+k}^q)\),
it processes all jobs~\(J_i\)
such that \(s_{iq}\)~is red.
By the choice of~\(s_{iq}\) for red cells, the machines have enough time to go around~\(H\) a second time. 
The length of the schedule is $n+2c(\HC)$:
each machine uses \(2c(\HC)\)~time for traveling,
\(n\)~time for processing the \(n\)~jobs,
and is never idle.
\qed\end{proof}

\noindent
\looseness=-1
The machines in the proof of
\cref{thm:artem} visit vertices twice.
The following example shows a \ROSUPT{} instance
for which \cref{thm:artem} computes an optimal schedule
and where the machines in an optimal schedule
\emph{have} to visit vertices repeatedly.
\begin{example}\label[example]{ex:badhamilton}
  Consider a \ROSUPT{} instance on a network~\(\Gr=(V,E,c,v^*)\)
  with two vertices~\(v^*=v_1\) and~\(v_2\)
  and one edge~\(e=\{v_1,v_2\}\) with \(c(e)=t\).
  Vertex~\(v^*\) contains one job~\(J_1\),
  vertex~\(v_2\) contains \(n-1\)~jobs~\(J_2,\dots,J_n\).
  A machine visiting~\(v_2\) only once
  either has to process first \(J_1\) and then \(J_2,\dots,J_n\),
  or first all of \(J_2,\dots,J_n\) and then~\(J_1\).
  
  Assume that we have \(m=m_1+m_2\)~machines,
  where a set~\(\Ms_1\) of \(m_1\)~machines processes~\(J_1\) first
  and a set~\(\Ms_2\) of \(m_2\)~machines processes~\(J_1\) last.
  Then one of the machines in~\(\Ms_1\) has to wait
  for the \(m_1-1\)~other machines in order to start~\(J_1\).
  Similarly, after finishing all jobs~\(J_2,\dots,J_n\),
  one of the machines in~\(\Ms_2\) has to wait
  for the \(m_2-1\)~other machines in order to start~\(J_1\).
  Thus, this schedule has makespan at least~\(2t+n+m/2-1\):
  there is a machine that spends
  \(2t\)~time for travelling, \(n\)~time for processing,
  and at least \(\max\{m_1-1,m_2-1\}\geq m/2-1\)~time for idling.

  \cref{thm:artem} gives a schedule with makespan
  \(4t+\max\{n,m\}\), which is smaller than \(2t+n+m/2-1\)
  if \(n\geq m\geq 4t-4\).
  Thus, in this instance, at least one machine in an optimal schedule
  has to visit \(v_2\)~twice, which incurs a travel time of~\(4t\).
  Since this machine also has to process \(n\)~jobs,
  it follows that the bound \(4t+\max\{n,m\}=4t+n\) given by \cref{thm:artem}
  is optimal in this case.
\end{example}

\noindent The above example shows that,
in an optimal solution to \ROSUPT{},
it can be necessary that machines visit vertices several times.
The following lemma gives and upper bound
on the number and length of stays of a machine in an optimal schedule.

\begin{lemma}\label[lemma]{lem:vislength}\label[lemma]{lem:shortcyc}
  Let \(S\)~be an optimal schedule
  for a \ROSUPT{} instance on a network~\(\Gr=(V,E,c,v^*)\).
  Let \(L\)~be the makespan of~\(S\) and \((R^q)_{M_q\in\Ms}\)~be routes
  of length at most~\(L\) compatible to~\(S\).
  Then, for each machine~\(M_q\),
  the route \(R^q=((a_k^q,v_k^q,b_k^q))_{k=1}^{s_q}\)
  \begin{enumerate}[(i)]
  \item has at most \(s_q\leq\pvstays\) stays, and
    
  \item the total length of the stays
    in any vertex~\(v\in V\) with \(n_v\)~jobs is
  \begin{align}
    \smashoperator{\sum_{1\leq k\leq s_q, v_k^q=v}}(b_k^q-a_k^q)\leq n_v+m-1.\label{slen}
  \end{align}
\end{enumerate}
\end{lemma}
\begin{proof}
  Let \(H\)~be a minimum-cost Hamiltonian cycle in~\(\Gr\).
  We show that, if one of (i) or (ii) is violated,
  then \(S\)~has makespan at least \(c(H)+m+n\),
  contradicting \cref{obs:lobo}.

  (i)
  If \(R^q\) had at least \(2g+m-2\)~stays,
  then,
  by \cref{prop:longcyc} in our graph-theoretic \cref{appendix},
  machine~\(M_q\) would be traveling
  for at least \(c(H)+m\)~time.
  Since the machine~\(M_q\) is processing jobs
  for \(n\)~time units,
  the makespan of \(S\) is at least \(c(H)+m+n\),
  a contradiction.
  
  (ii)
  Machine~\(M_q\) takes at least \(c(H)+n\)~time just
  for visiting all vertices and processing all jobs.
  If \eqref{slen} does not hold,
  then machine~\(M_q\) is neither travelling nor processing
  for at least \(m\)~time units.
  Thus, the length of route~\(R^q\) and,
  therefore,
  the makespan of \(S\),
  is at least \(c(H)+n+m\), again a contradiction.
  \qed
\end{proof}

\section{Fixed-parameter algorithm}\label{sec:fpt}
\noindent 
In this section, we present a fixed-parameter algorithm for \ROSUPT{}.
The following simple algorithm shows
that our main challenge will be ``bottleneck vertices'' that
contain less jobs than there are machines.

\begin{proposition}\label[proposition]{thm:trivsched}
  \ROSUPT{} is solvable in \(O(2^gg^2+mn)\)~time if each vertex
  contains at least \(m\)~jobs.
\end{proposition}
\begin{proof}
  First, compute a minimum-cost Hamiltonian cycle~\(\HC=(v_1,v_2,\dots,v_{\nverts},v_{\nverts+1})\)
  in the input network~\(\Gr=(V,E,c,v^*)\) in \(\bigO(2^{\nverts}\cdot\nverts^2)\)~time
  using the algorithm of \citet{Bell62}, \citet{HelK62},
  where \(v_1=v_{\nverts+1}=v^*\)~is the depot.

  Denote by \(n_k\)~the number of jobs in~\(v_k\).
  Each machine will follow the same route~$R$
  of $\nverts+1$~stays $(a_1,v_1,b_1),\dots,\allowbreak (a_{\nverts+1},v_{\nverts+1},b_{\nverts+1})$,  where \(a_1:=0\), \(b_{\nverts+1}:=a_{\nverts+1}\),
  \begin{align*}
    b_{k}&:=a_{k}+n_k&\text{ and }&&
    a_{k+1} &:= b_{k}+c(v_k,v_{k+1}) &\text{for $k\in\{1,\dots,\nverts\}$}.
  \end{align*}
  That is, all machines will stay in vertex~\(v_k\) for the same \(n_k\)~time units.
  Now consider the schedule~\(S\) that
  schedules each job~\(J_i\) in a vertex~\(v_k\)
  on each machine~\(M_q\) at time
  \[
    S(J_i,M_q)=a_k+(i-q)\bmod n_k.
  \]
  Since \(n_k\geq m\), it is easy to verify
  that \(S(J_i,M_q)\ne S(J_j,M_r)\)
  for \(i=j\) and \(q\ne r\)
  and for \(q=r\) and \(i\ne j\) if \(J_i\) and \(J_j\) are in the same vertex~\(v_k\)
  (we did this for \eqref{eq:shift} in the proof of \cref{thm:artem}).
  For a job \(J_i\) in a vertex~\(v_k\)
  and a job \(J_j\) in a vertex~\(v_\ell\)
  with \(k< \ell\), \(S(J_i,M_q)\ne S(J_j,M_q)\)
  easily follows 
  from \(S(J_i,M_q)< a_k+n_k\leq a_\ell\leq S(J_j,M_q)\).

  It is obvious that \(S\)~is compatible to route~\(R\).
  The length of~\(R\) is \(c(H)+n\), which is optimal by \cref{obs:lobo}.
  \qed
\end{proof}

\noindent
For the general \ROSUPT{} problem, we prove the following theorem,
which is our main algorithmic result.

\begin{theorem}\label{thm:ros-fpt}
  \ROSUPT{} is solvable in \(2^{\bigO(\nverts\nmach^2\log\nverts\nmach)}+O(mn\log n)\text{ time.}\)
\end{theorem}

\noindent
In view of \cref{thm:trivsched} and \cref{ex:badhamilton},
the main challenge for our algorithm
are ``bottleneck vertices'' with few jobs,
which may force machines to idle.
\begin{definition}[critical vertices, \boldmath\(n_v\)]\label[definition]{def:crit}
  For a vertex~$v$ in the network of a \ROSUPT{} instance, let \(n_v\)~denote the number of jobs in~\(v\).   A vertex~\(v\) is \emph{critical} if $\jobsin{v}<m$.
\end{definition}

\noindent
To handle critical vertices,
we exploit that,
by \cref{lem:vislength},
the \(m\)~routes compatible to an optimal schedule
together have at most \(\totstays\)~stays
and stay in critical vertices last at most \(2m-1\)~time.
In time that depends only on~\(g\) and~\(m\), 
we can thus try all possibly optimal chronological sequences
of stays of machines in vertices,
lengths of stays in critical vertices,
and time differences between stays in critical vertices
if they intersect.
Thus, we can essentially try all possibilities of fixing
everything in the routes
except for the exact arrival and departure times of stays.
For each such possibility,
we will try to compute the arrival and departure times
using integer linear programming,
a schedule in critical vertices using brute force,
and a schedule in uncritical vertices using edge colorings of bipartite graphs.

In the following, we first formalize
these partially fixed schedules and routes
and then give a description of our algorithm in pseudo-code.

\begin{definition}[critical schedule]
  A \emph{critical schedule} is a
  function~$\sched\colon \Js{}\times\Ms{}\to\mathbb N\cup\{\bot\}$
  determining the \emph{start time}~$\sched(J_i, M_q)\ne\bot$
  of each job~$J_i$ in a critical vertex on each machine~$M_q$
  and having \(S(J_i,M_q)=\bot\) for all jobs~\(J_i\) in non-critical vertices.

  A critical schedule has to satisfy \cref{def:ros}\eqref{def:feas2}--\eqref{def:feas4}
  for all jobs~\(J_i\) and machines~\(M_q\) with \(S(J_i,M_q)\ne\bot\).
\end{definition}

\noindent
We now formally define \emph{pre-schedules},
which fix routes up to the exact arrival and departure times of stays.
The definition is illustrated in \cref{fig:presc}:

\begin{figure*}
  \centering
  \includegraphics{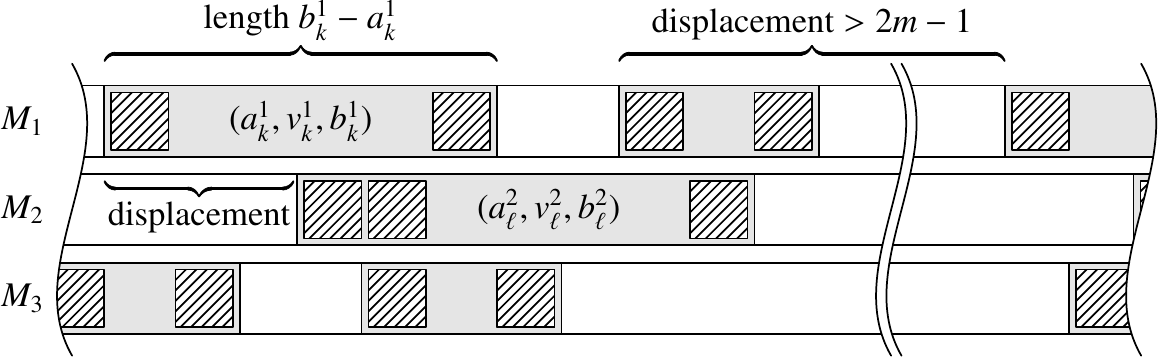}
  \caption{Shown is a part of a schedule for three
    machines~\(M_1,M_2\), and~\(M_3\).
    The horizontal axis is time.
    Uniformly gray boxes are stays in critical vertices that comply with some pre-schedule (the pre-schedule is not shown).
    Hatched squares correspond to jobs being processed.
    Illustrated are the lengths of stays and displacements between stays in critical vertices that are consecutive in the pre-stay sequence
    (stays in non-critical vertices are not shown).
    Herein, this displacement is either smaller
    than~\(2\nmach\),
    in which case the time difference between the stays
    is fixed for \emph{any} complying set of routes,
    or at least \(2\nmach\), in which case the stays cannot
    intersect in time for \emph{any} set of complying routes
    since stays in critical vertices have length at most
    \(2\nmach-1\) by \cref{def:presc}\eqref{la1}.}
  \label{fig:presc}
\end{figure*}
\begin{definition}[pre-schedule]\label[definition]{def:presc}
  A \emph{pre-schedule} is a triple~\((T,A,D)\).  Herein,
  \begin{itemize}[\(T\)]
  \item[\(T\)] is a \emph{pre-stay sequence}~\((T_k)_{k=1}^s\) with \(s\leq\totstays\)
    and will fix a chronological order of all machine stays
    (by non-decreasing arrival times).
    The \emph{\(k\)-th pre-stay}~\(T_k=(q_k,w_k)\)
    will require machine~\(M_{q_k}\) to stay in vertex~\(w_k\).
  \end{itemize}
  For the definition of the components~\(A\) and~\(D\) of a pre-schedule, let
  \(\mathcal K\subseteq\{1,\dots,s\}\) be the indices~\(k\) of pre-stays~\((q_k,w_k)\)
  such that \(w_k\)~is critical.  Then,
  \begin{itemize}[\(A\ \)]
  \item[\(A\colon\mathcal K\to\{0,\dots,2m-1\}\)] is
    called \emph{length assignment}
    and will fix the length of the stay corresponding
    to the \(k\)-th pre-stay to be~\(A(k)\), and
  \item[\(D\colon\mathcal K\to\{0,\dots,2m\}\)] is
    called \emph{displacement} and
    will fix the time difference
    between the stays corresponding to the \(k\)-th pre-stay
    and the previous pre-stay in a critical vertex
    to be \(D(k)\)  if \(D(k)\leq 2m-1\)
    or to be at least \(D(k)\) if \(D(k)=2m\) (which means
    that it will prevent the two stays from intersecting).
  \end{itemize}
  We now formalize routes that \emph{comply} with a pre-schedule.
  To this end, we denote by
  \begin{itemize}[\(\sigma_{T,k}\)]
  \item[\(\sigma_{T,k}\)]  the number %
    such that \((q_k,w_k)\) for some vertex~\(w_k\) is
    the \(\sigma_{T,k}\)-th pre-stay
    of machine~\(M_{q_k}\) in the pre-stay sequence~\(T\).
    We omit the subscript~\(T\) if the pre-stay sequence is clear from context.
  \end{itemize}
  Routes~\((R^{q})_{M_q\in\Ms}\), where \(R^q=((a_k^q,v_k^q,b_k^q))_{k=1}^{s_q}\), \emph{comply} with the \ps{} sequence~\(T\) if and only if
\begin{enumerate}[(i)]
\item\label{presc1} \(v^q_{\sigma_k}=w_k\) for all \(k\in\{1,\dots,s\}\), that is, \(M_q\) makes its \(\sigma_k\)-th stay in~\(w_k\), and
\item\label{presc2} for \ps{}s~\(({q_k},w_k)\) and~\(({q_{\ell}},w_{\ell})\) with \(k<\ell\),
  one has~\(a_{\sigma_k}^{q_k}\leq a_{\sigma_{\ell}}^{q_\ell}\), that is,
  the stays in all routes~\((R^{q})_{M_q\in\Ms}\) are chronologically ordered according to~\(T\).
\end{enumerate}
  Routes~\((R^q)_{M_q\in\Ms}\)  \emph{comply} with a \las~\(A\) if,
\begin{enumerate}[(i)]
\setcounter{enumi}{2}
\item\label{la1} each \ps{}~\((q_k,w_k)\) in~\(T\) with~\(k\in\mathcal K\)
has length \(b_{\sigma_k}^{q_k}-a_{\sigma_k}^{q_k}=A(k)\).
\end{enumerate}
Routes~\((R^q)_{M_q\in\Ms}\)  \emph{comply} with a \dis~\(D\) if
\begin{enumerate}[(i)]
\setcounter{enumi}{3}
\item\label{dis1} for two \ps{}s~\(({q_k},v_k)\) and~\(({q_\ell},v_\ell)\) such that \(k<\ell\), \(\{k,\ell\}\subseteq\mathcal K\), and \(i\notin\mathcal K\) for each \(i\in\{k+1,\dots,\ell-1\}\), one has
  \begin{align*}
    a_{\sigma_\ell}^{q_\ell}&\geq a_{\sigma_k}^{q_k}+2m&&\text{if~\(D(j)=2\nmach\) and}\\
    a_{\sigma_\ell}^{q_\ell}&=a_{\sigma_k}^{q_k}+D(\ell)&&\text{if~\(D(j)<2\nmach\).}
  \end{align*}
\end{enumerate}
Routes \emph{comply} with \((T,A,D)\) if they comply with each of~\(T\), \(A\), and~\(D\).
\end{definition}

\noindent
\cref{lem:vislength} implies that there is a pre-schedule
complying with some routes compatible
to an optimal schedule.
We thus enumerate all possible pre-schedules and, for each,
try to find routes complying with
the pre-schedule and compatible to an optimal schedule.
This leads to the following algorithm.
\begin{algorithm}\small\leavevmode
  \label[algorithm]{alg:outline}
  \begin{compactdesc}
  \item[\it Input:] A \ROSUPT{} instance~\(I=(\Gr,\Js,\Ms,\loc,P)\) on a network~\(\Gr=(V,E,c,v^*)\).
    
  \item[\it Output:] A minimum-makespan schedule for~\(I\).
  \end{compactdesc}
  \begin{compactenum}
  \item Preprocess \(G\)~to establish the triangle inequality.\hfill// \cref{lem:metric}\label{lin:metric}

  \item \(H\gets{}\)minimum-cost Hamiltonian
    cycle~\(H\) in~\(G\).\label{lin:hamcyc}

  \item \textbf{for} \(L=c(H)+n\) \textbf{to} \(c(H)+n+m-1\) \textbf{do}
    \hfill// Try to find schedule with makespan~\(L\) (Obs.~\ref{obs:lobo})\label{lin:chooseL}

  \item \label{lin:preschloop}\quad\textbf{foreach} pre-schedule
    \((T,A,D)\) \textbf{do}\hfill// \cref{numprescheds}

  \item \qquad\textbf{if} there are routes~\((R^q)_{M_q\in\Ms}\) that
    comply with~\((T,A,D)\), \hfill// \cref{lem:complroutes}\label{lin:consroute}

    \qquad\qquad that     have length at most~\(L\) each, and

    \qquad\qquad stay in each \emph{non-critical} vertex~\(v\in V\) at least \(n_v\)~time,
    \textbf{then}
    
  \item\quad\qquad\textbf{if} there is a critical
    schedule~\(S'\) compatible with~\((R^q)_{M_q\in\Ms}\), \textbf{then}
    \hfill// 
    \cref{lem:enumscheds}\label{lin:conspartsched}

  \item
    \qquad\qquad complete~\(S'\) into a feasible schedule \(S\) compatible with~\((R^q)_{M_q\in\Ms}\).
          \hfill// \cref{lem:brf}\label{lin:complsched}
          
  \item\qquad\qquad \textbf{return} $S'$.\label{lin:retsched}
\end{compactenum}
\end{algorithm}

\noindent
To prove the correctness and the running time of \cref{alg:outline},
we prove the lemmas named in its comments.
We already proved \cref{lem:metric,obs:lobo}
and we continue proving lemmas in the order appearing in the algorithm.
First,
we bound the number of pre-schedules
and thus
the number of repetitions of the loop in line~\ref{lin:preschloop}.

\begin{lemma}\label[lemma]{numprescheds}
  In line~\ref{lin:preschloop} of \cref{alg:outline},
  there are at most \(2^{O(m^2\log gm+mg\log gm)}\) pre-schedules,
  which can be enumerated in \(2^{O(m^2\log gm+mg\log gm)}\)~time.
\end{lemma}
\begin{proof}
  There are at most \((mg)^\totstays\) pre-stay sequences:
  a pre-stay sequence consists of at most \(\totstays\) pre-stays,
  each of which is a pair of one of \(m\)~machines and one of \(g\)~vertices.
  For each pre-stay sequence, there are at most \((2m)^\totstays\) length assignments
and at most \((2m+1)^\totstays\) displacements.
Thus, the number  of pre-schedules is at most
\((4gm^3+2gm^2)^\totstays=2^{(\totstays)\log (4gm^3+2gm^2)}\).
They can obviously be enumerated in the stated running time
using a recursive algorithm.
  \qed
\end{proof}

\noindent
We realize the check in line~\ref{lin:consroute}
by testing
the feasibility of an integer linear program
whose number of variables, number of constraints, and
absolute value of coefficients is bounded by
\(\bigO(\totstays)\).
By Lenstra's theorem, this works in \(2^{O(m^2\log gm+mg\log gm)}\)~time:
\begin{proposition}[\citet{Len83}; see also \citet{Kan87}]\label[proposition]{thm:lenstra}
  A feasible solution to an integer linear program of size~\(n\) with \(p\)~variables
  is computable using \(p^{\bigO(p)}\cdot n\)~arithmetical operations,
  if such a feasible solution exists.
\end{proposition}

\begin{lemma}\label[lemma]{lem:complroutes}
  The  routes in line~\ref{lin:consroute} of \cref{alg:outline}
  are computable in \(2^{O(m^2\log gm+mg\log gm)}\)
  time, if they exist.
\end{lemma}
\begin{proof}
  Let \(T:=(({q_k},w_k))_{i=1}^s\) be the pre-stay sequence
  in the pre-schedule~\((T,A,D)\) enumerated in line~\ref{lin:preschloop}.
  By \(s_q\), denote the number of pre-stays of a machine~\(M_q\) in~\(T\).
  We compute the routes \((R^q)_{M_q\in\Ms}\),
  where \(R^q:=((a_k^q,v_k^q,b_k^q))_{k=1}^{s_q}\), as follows.
For each pre-stay \(({q_k},w_k)\) on~\(T\),
we let~\(v_{\sigma_k}^{q_k}:=w_k\).  If \(v_1^q\ne v^*\) or \(v^*\ne w_{s_q}^q\)
for some machine~\(M_q\), where \(v^*\)~is the depot, then there are no routes
complying with~\((T,A,D)\) and  we return ``no'' accordingly.

Otherwise, by \cref{def:presc}, the \(a_k^q\) and \(b_k^q\)
for all machines~\(M_q\) and \(1\leq k\leq s_q\) together
are at most \(2\nmach\cdot(\pvstays)\)~variables.
They can be determined by a feasible solution to an integer linear program.
This, together with \cref{thm:lenstra} directly yields the running time stated in \cref{lem:complroutes}. 
The linear program consists of the following constraints.
  We want each route to have length at most~\(L\), that is,
  \begin{align*}
    b_{s_q}^q&\leq L&&\text{for each \(M_q\in\Ms\)}.\\
    \intertext{The difference between departure and arrival times are the travel times, that is,}
    b_{k}^q+c(v_{k}^q,v_{k+1}^q)&= a_{k+1}^q&&\text{for each \(M_q\in\Ms\) and \(1\leq k\leq s_q-1\).}\\
\intertext{Stays should have non-negative length, that is,}
    a_{k}^q&\leq b_{k}^q&&\text{for each \(M_q\in\Ms\) and \(1\leq k\leq s_q\).}\\
    \intertext{Each machine should stay in~\(v\in V\) for at least \(n_v\)~time, that is}
    \sum_{\substack{1\leq k\leq s_q\\w_{k}^{q}=v}}(b_k^q-a_k^q)&\geq n_v&&\text{for each~\(M_q\in\Ms\) and~\(v\in V\).}
\intertext{Stays must be ordered according to the pre-stay sequence~\(T\), that is}
    a_{\sigma_k}^{q_k}&\leq a_{\sigma_\ell}^{q_\ell}&&\text{for pre-stays \(({q_k},v_k)\) and \(({q_\ell},v_\ell)\) with \(k\leq \ell\).}
    \intertext{Stays should adhere to the length assignment~\(A\), that is}
    b_{\sigma_k}^{q_k}-a_{\sigma_k}^{q_k}&=A(k)&&\text{for each pre-stay \(({q_k},v_k)\) such that \(v_k\)~is critical.}
    \intertext{Finally, routes have to comply with the displacement~\(D\).  To formulate the constraint, let \(\mathcal K\) be the indices of \ps{}s of~\(T\) in critical vertices.  For any two pre-stays \(({q_k},w_k)\) and \(({q_\ell},w_\ell)\) with \(k<\ell\), \(\{k,\ell\}\subseteq\mathcal K\), and \(i\notin\mathcal K\) for each \(i\in\{k+1,\dots,\ell-1\}\), we want}
a_{\sigma_\ell}^{q_\ell}&\geq a_{\sigma_k}^{q_k}+D(\ell)&&\text{if \(D(\ell)=2\nmach\), and}\\
a_{\sigma_\ell}^{q_\ell}&=a_{\sigma_k}^{q_k}+D(\ell)&&\text{if \(D(\ell)<2\nmach\).}\hspace{5cm}\qed
  \end{align*}
\end{proof}

\noindent
Next, we show how to realize the check in \ref{lin:conspartsched} of \cref{alg:outline}.

\begin{lemma}\label[lemma]{lem:enumscheds}
  The critical schedule~\(S'\) in line~\ref{lin:conspartsched}
  is computable
  in \(2^{\bigO(\nverts\nmach^2\log\nmach)}\)~time,
  if such a schedule exists.
\end{lemma}
\begin{proof}
  In total, there are at most \(\nverts\cdot\nmach\)~jobs in critical vertices.
  Thus, we determine the starting time~\(S(J_i,M_q)\)
  for at most \(\nverts\cdot\nmach^2\)~pairs \((J_i,M_q)\in\Js\times\Ms\).
  By \cref{lem:vislength},
  each machine can process all of its jobs in a critical vertex staying there
  no longer than \(2\nmach-1\)~units of time.
  Thus, for each of at most \(\nverts\cdot\nmach^2\)~pairs \((J_i,M_q)\),
  we enumerate all possibilities of choosing \(S(J_i,M_q)\) among the
  smallest \(2\nmach-1\)~time units where~\(M_q\) stays in vertex~\(\loc(J_i)\).

  There are \((2\nmach-1)^{\nverts\cdot\nmach^2}\in 2^{\bigO(\nverts\nmach^2\log\nmach)}\)
  possibilities to do so.  The feasibility's of each variant
  can be checked in \(O(g\cdot m^2)\)~time and
  they can all be enumerated in \(2^{\bigO(\nverts\nmach^2\log\nmach)}\)~total time:
  since the routes~\((R^q)_{M_q\in\Ms}\) comply with~\((T,A,D)\),
  they
  have at most \(\totstays\)~stays in total,
  thus we can list the first \(2m-1\)~time units
  that a concrete machine stays in a concrete vertex
  in \((2m-1)+\totstays\)~time.
\qed\end{proof}

\noindent
In the following, we provide the last building block
for proving the correctness and the running time of \cref{alg:outline}:
\cref{lem:vislength} already shows that there is a pre-schedule
that complies with some routes that are compatible with
an optimal schedule.
Thus, it is sufficient to try, for each pre-schedule~\((T,A,D)\),
to search for schedules compatible with routes
complying with~\((T,A,D)\).
However, the algorithm only searches for schedules
compatible to \emph{one} collection of machine routes
complying with~\((T,A,D)\).
The following lemma shows that this is sufficient.

\begin{lemma}\label[lemma]{lem:brf}
  If a \ROSUPT{} instance on a network~\(\Gr=(V,E,c,v^*)\)
  allows for a feasible schedule
  compatible to routes
  complying with a pre-schedule~\((T,A,D)\), then
  \begin{enumerate}[(i)]
  \item\label{brf:psched} for any collection of routes~\((R^q)_{M_q\in\Ms}\) complying with~\((T,A,D)\),
    there is a critical schedule compatible with \((R^q)_{M_q\in\Ms}\), and
  \item\label{brf:completable} any such critical schedule
    can be extended into a feasible schedule compatible to~\((R^q)_{M_q\in\Ms}\)
    if each route~\(R^q\) stays in each vertex~\(v\)
    for at least~\(n_v\)~time.
  \end{enumerate}
  Moreover, line~\ref{lin:complsched} can be carried out in \(O(m^2+gm+mn\log n)\)~time.
\end{lemma}
\begin{proof}
  \eqref{brf:psched}
  Let \(S^*\)~be a feasible schedule compatible to
  some routes~\((R^{q*})_{M_q\in\Ms}\)
  complying with~\((T,A,D)\).  Denote the pre-stay sequence~\(T=(({q_k},w_k))_{k=1}^s\).
  We show how to construct a critical schedule~\(S\)
  compatible with respect to the routes~\((R_{M_q})_{M_q\in\Ms}\).
  Denote \(\smash{R^{q*}=((a^{q*}_k,v^{q*}_k,b^{q*}_k))_{k=1}^{s_q^*}}\),
  and \(R^q=((a_k^q,v_k^q,b_k^q))_{k=1}^{s_q}\)
  for each machine~\(M_q\).
  By \cref{def:presc}\eqref{presc1} and \cref{def:ros}\eqref{def:feas4}, for each job~\(J_i\) and machine~\(M_q\) there is an index \(P(J_i,M_q):=p\) of a pre-stay~\(({q_p},w_p)=(q,\loc(J_i))\) on~\(T\) such that
  \begin{align}
    a^{q*}_{\sigma_p}\leq S^*(J_i,M_q)< b^{q*}_{\sigma_p}.\label{leqeq}
  \end{align}
  Since the routes \((R^{q*})_{M_q\in\Ms}\) and \((R^q)_{M_q\in\Ms}\) comply with~\(T\), by \cref{def:presc}\eqref{presc1}, one has \(s_q=s^*_q\) and, moreover, \(v_k^q=v_k^{q*}\) for each machine~\(M_q\) and \(1\leq k\leq s_q\).
For each job~\(J_i\) and machine~\(M_q\), we define
  \begin{align*}
    S(J_i,M_q):=
    \begin{cases}
      \bot&\text{if \(\loc(J_i)\)~is not critical},\\
      S^*(J_i,M_q)-a^{q*}_{\sigma_p}+a^{q}_{\sigma_p}&\text{otherwise,}
    \end{cases}
  \end{align*}
where \(p=P(J_i,M_q)\).

We show that \(S\)~is a critical schedule compatible with the routes \((R^q)_{M_q\in\Ms}\).
For each job~\(J_i\) in a critical vertex and each machine~\(M_q\),
we first show that machine~\(M_q\) stays in~\(\loc(J_i)\) when processing job~\(J_i\).
More precisely, for  \(p=P(J_i,M_q)\),
we show \(a^{q}_{\sigma_p}\leq S(J_i,M_q)<b^{q}_{\sigma_p}\) as follows.
By adding \(-a^{q*}_{\sigma_p}+a^{q}_{\sigma_p}\) to both sides of
\[  a^{q*}_{\sigma_p}\leq S^*(J_i,M_q),\]
which holds since \(p\)~is chosen so as to satisfy \eqref{leqeq}, one gets
\[  a^{q}_{\sigma_p}\leq S^*(J_i,M_q)-a^{q*}_{\sigma_p}+a^{q}_{\sigma_p}=S(J_i,M_q).\]
 Moreover, since both \(R^{q*}\) and \(R^q\) comply with the length assignment~\(A\), by \cref{def:presc}\eqref{la1}, one has \(b_{\sigma_k}^{q_k}-a_{\sigma_k}^{q_k}=A(k)=b^{q_k*}_{\sigma_k}-a^{q_k*}_{\sigma_k}\) for all pre-stays \(({q_k},w_k)\) such that \(w_k\)~is critical. Thus, by adding \(-a^{q*}_{\sigma_p}+a^{q}_{\sigma_p}\) to both sides of
\[S^*(J_i,M_q)< b^{q*}_{\sigma_p},\]
which holds since \(p\)~is chosen so as to satisfy \eqref{leqeq}, one gets
\[  S(J_i,M_q)= S^*(J_i,M_q)-a^{q*}_{\sigma_p}+a^{q}_{\sigma_p}< b^{q*}_{\sigma_p}-a^{q*}_{\sigma_p}+a^{q}_{\sigma_p}=b^{q}_{\sigma_p}.\]
It remains to show that~\(S^*\) processes no two jobs at the same time and that no two machines process one job at the same time.  To this end, consider jobs~\(J_i,J_j\) in critical vertices and machines~\(M_q, M_r\).  If \emph{either} \(i=j\) or \(q=r\), then \(S^*(J_i,M_q)\ne S^*(J_j,M_r)\).  Thus, it is sufficient to show that \(S^*(J_i,M_q)\ne S^*(J_j,M_r)\) implies \(S(J_i,M_q)\ne S(J_j,M_r)\).  
To this end, let \(p:=P(J_i,M_q)\) and \(\pi:=P(J_j,M_r)\).
Without loss of generality, assume that~\(p\leq\pi\).  
By \cref{def:presc}\eqref{presc2}, one has \(a^{q}_{\sigma_p}\leq a^{r}_{\sigma_\pi}\). 

If \(b_{\sigma_p}^{q}\leq a^{r}_{\sigma_\pi}\), then \(S(J_i,M_q)\ne S(J_r,M_r)\) follows from
\(
S(J_i,M_q)< b^{q}_{\sigma_p}\leq a^{r}_{\sigma_\pi}\leq S(J_j,M_r).
\)
Otherwise, since \(b^{q}_{\sigma_p}-a^{q}_{\sigma_p}=A(p)\leq 2\nmach-1\) by \cref{def:presc}\eqref{la1}, one has \(a^r_{\sigma_\pi}-a^q_{\sigma_p}\leq b^q_{\sigma_p}-a^q_{\sigma_p}\leq 2\nmach-1\).  Thus, for \(\mathcal K(p,\pi]:=\{k\in\{p+1,\dots,\pi\}\mid ({q_k},w_k)\) is a pre-stay of~\(T\) in a critical vertex\(\}\),  one has, by \cref{def:presc}\eqref{presc2} and \eqref{dis1},
\begin{align}
a^r_{\sigma_\pi}-a^q_{\sigma_p}=\smashoperator{\sum_{k\in\mathcal K(p,\pi]}}D(k)=a^{r*}_{\sigma_\pi}-a_{\sigma_p}^{q*}\label{equality}
\end{align}
since both tours~\((R^q)_{M_q\in\Ms}\) and~\((R^{q*})_{M_q\in\Ms}\)
comply with the displacement~\(D\).  
By adding \(a^r_{\sigma_\pi}-a^{q*}_{\sigma_p}\) to both sides of
\begin{align*}
  S(J_i,M_q)+a^{q*}_{\sigma_p}-a^q_{\sigma_p}&=S^*(J_i,M_q)\ne S^*(J_j,M_r)=S(J_j,M_r)+a^{r*}_{\sigma_\pi}-a^r_{\sigma_\pi},
\end{align*}
which is true by the definition of~\(S\) from~\(S^*\), one obtains
\[
S(J_i,M_q)+a^r_{\sigma_\pi}-a^q_{\sigma_p}\ne S(J_r,M_r)+a^{r*}_{\sigma_\pi}-a^{q*}_{\sigma_p},
\]
and, therefore,  \(S(J_i,M_q)\ne S(J_j,M_r)\) from \eqref{equality}.

\eqref{brf:completable}
We complete any critical schedule~\(S'\) compatible
with the routes~\((R^q)_{M_q\in\Ms}\) into a feasible schedule
compatible with \((R^q)_{M_q\in\Ms}\) as follows.

For each machine~\(M_q\)
and each non-critical vertex~\(v\),
let \(T_{qv}\subseteq\mathbb N\)~be a set of \(n_v\)~arbitrary times
where machine~\(M_q\) stays in~\(v\) according to
route~\(R^q\) and let \(T_v:=\bigcup_{M_q\in\Ms} T_{qv}\).
For each vertex~\(v\in V\), create a bipartite
graph~\(G_{v}=(\Ms\cup T_v,E_{v})\),
where \(E_{v}\)~contains an edge~\(\{M_q,t\}\) between a machine~\(M_q\) and a time \(t\in T_{qv}\)
if and only if \(M_q\) is in~\(v\) at time~\(t\).
Each vertex of~\(G_v\) has degree~\(n_v\) or~\(m\), where \(m\leq n_v\)
since \(v\)~is non-critical.
Thus, \(G_v\)~allows for a proper edge coloring~\(c_v\colon E_{qv}\to \{1,\dots,n_v\}\):
if one of \(q=r\) or \(t=t'\), then \(c_v(\{M_q,t\}\ne c_v(\{M_r,t'\})\).
This coloring will tell us which job machine~\(M_q\) will process at time~\(t\).
Let the jobs in each vertex~\(v\)
be~\(\{J_1^v,\dots,J_{n_v}^v\}\).
Then, for any machine~\(M_q\) and job~\(J_j^v\)
in a non-critical vertex~\(v\),
there is a unique~\(t_{qj}\) such that \(c_v(\{M_q,t_{qj}\})=j\).
We thus define our schedule~\(S\) as
\[
  S(J_i,M_q):=
  \begin{cases}
    S'(J_i,M_q)&\text{ if $J_i$ is in a critical vertex}\\
    t_{qj}&\text{ if $J_i=J_j^v$ for some non-critical vertex~\(v\)}.
  \end{cases}
\]
By construction from schedule~\(S'\) for critical vertices,
which is compatible to the routes~\((R^q)_{M_q\in\Ms}\),
and
from the edge-coloring~\(c_v\) for non-critical vertices~\(v\),
schedule \(S\)~is compatible to~\((R^q)_{M_q\in\Ms}\).
Moreover, from this, \(S(J_i,M_q)\ne S(J_j,M_q)\)
follows if \(J_i\ne J_j\) are in the same vertex.
If \(J_i\ne J_j\) are in different vertices,
then this follows from the compatibility of~\(S\)
with the routes~\((R^q)_{M_q\in\Ms}\): machine~\(M_q\)
cannot stay in two vertices at the same time.
Finally, \(S(J_i,M_q)\ne S(J_i,M_r)\) for \(M_q\ne M_r\)
follows from \(S'(J_i,M_q)\ne S'(J_i,M_r)\) if \(J_i\)~is in a critical vertex.
If \(J_i\)~is a job in a non-critical vertex, say \(J_j^v\),
then \(S(J_i,M_q)=S(J_j,M_r)\) implies \(t_{qj}=t_{rj}\) and,
in turn, \(c_v(\{M_r,t_{qj}\})=j=c_v(\{M_r,t_{rj}\})\),
contradicting the fact that \(c_v\)~is a proper edge coloring.

We analyze the running time for this completion step.
Since the routes~\((R^q)_{M_q\in\Ms}\) have at most
\(\totstays\)~stays, one can compute the sets~\(T_{vq}\)
for all vertices~\(v\) and machines~\(M_q\) in
\(O(\totstays+n)\)~time.
For each vertex~\(v\),
the bipartite graph~\(G_{v}\) can be generated
in \(O(mn_v)\)~time and an edge-coloring into \(n_v\)~colors
can be computed in \(O(mn_v\log n_v)\)~time~\citep{COS01}.
Thus, in total,
we can compute schedule~\(S\) in
\begin{align*}
  O(\totstays+n)+\sum_{v\in V}O(mn_v\log n_v)=O(m^2+gm+mn\log n)\text{\quad time.}\qquad\qed
\end{align*}
\end{proof}

\noindent
We can now prove the correctness and running time of \cref{alg:outline}.

\begin{proof}[of \cref{thm:ros-fpt}]
  Let \(L^*\)~be the makespan of an optimal schedule~\(S^*\)
  for the \ROSUPT{} instance input to \cref{alg:outline}.
  We only have to show that (and in which time) \cref{alg:outline} outputs
  a feasible schedule~\(S\) with makespan at most~\(L^*\).
  To this end, let \((R^{q*})_{M_q\in\Ms}\)~be routes compatible to~\(S^*\),
  each of length at most~\(L^*\).
  Line~\ref{lin:metric} of \cref{alg:outline}
  can be carried out in \(O(g^3)\) time
  using the Floyd-Warshall algorithm,
  line~\ref{lin:hamcyc} in \(O(2^gg^2)\)~time using
  the algorithm of \citet{Bell62}, \citet{HelK62}.  By
  \cref{obs:lobo}, \(L=L^*\) in at least one iteration of the loop in
  line~\ref{lin:chooseL}.
  We now consider this iteration.
  By \cref{numprescheds},  in line~\ref{lin:preschloop}, we enumerate
  \(2^{O(m^2\log gm+mg\log gm)}\)~pre-schedules.
  By \cref{lem:vislength}, among them there is a pre-schedule~\((T^*,A^*,D^*)\)
  that complies with~\((R^{q*})_{M_q\in\Ms}\).
  Since \(S^*\)~processes all jobs, the routes~\((R^{q*})_{M_q\in\Ms}\)
  stay in each vertex~\(v\) at least \(n_v\)~time.
  Thus, the test in line~\ref{lin:consroute} succeeds for~\((T^*,A^*,D^*)\) and \(L=L^*\)
  and, by \cref{lem:complroutes}, 
  can be carried out in \(2^{O(m^2\log gm+mg\log gm)}\)~time.
  By \cref{lem:enumscheds} the test in line~\ref{lin:conspartsched}
  can be carried out in \(2^{\bigO(\nverts\nmach^2\log\nmach)}\)~time
  and, by \cref{lem:brf}\eqref{brf:psched}, it succeeds. We get the
  the feasible schedule~\(S\) in line~\ref{lin:complsched}
  in \(O(m^2+gm+mn\log n)\)~time by \cref{lem:brf}\eqref{brf:completable}.
  Its makespan is at most~\(L=L^*\).
  The overall running time of the algorithm is \(O(g^3+2^gg^2)+m\cdot 2^{O(m^2\log gm+mg\log gm)}\cdot (2^{O(m^2\log gm+mg\log gm)}+ 2^{\bigO(\nverts\nmach^2\log\nmach)})+O(m^2+gm+mn\log n)\).\qed
\end{proof}

\section{Open questions}\label{sec:conclusion}
\noindent
We have shown that \ROSUPT{} is fixed-parameter tractable
with respect to the parameter~\(m+g\) and,
in the absence of critical vertices,
also with respect to the parameter~\(g\).
However, the question
whether \ROSUPT{} with critical vertices
and an unbounded number of machines
is polynomial-time solvable is open even for two vertices.

\begin{acknowledgements}
  We are thankful 
  to Mikhail Khachay for pointing out the work of \citet{Mad74}.
\end{acknowledgements}

\bibliographystyle{spbasic}
\bibliography{ros-upt}

\appendix
\section{On the minimum weight of long closed walks containing all vertices}
\label{appendix}
In the following,
we prove \cref{prop:longcyc},
which we used to prove \cref{lem:shortcyc}.
For its formal statement and proof,
we have to formally distinguish two different kinds of paths and cycles:

\begin{definition}[closed walks, cycles]
  Let \(G=(V,E)\)~be a multigraph with edge weights~\(c\colon E\to\mathbb N\).
  A~\emph{walk of length}~\(\ell\) in~\(G\)
  is an alternating sequence~\((v_1,e_1,v_2,e_2,\dots,v_\ell,e_\ell,v_{\ell+1})\)
  of vertices and edges
  such that
  \(v_i\)~and \(v_{i+1}\)~are the end points of~\(e_i\)
  for each \(i\in\{1,\dots,\ell\}\).
  Its \emph{weight} is \(\sum_{i=1}^{\ell}c(e_i)\), its
  \emph{internal vertices} are \(v_2,\dots,v_\ell\), and
  it is \emph{closed} if \(v_1=v_{\ell+1}\).

  \looseness=-1
  A graph is \emph{Eulerian} if it contains an \emph{Euler tour}---a closed walk
  that contains each edge of~\(G\) \emph{exactly} once.
  It is known that a connected undirected multigraph is Eulerian
  if and only if each vertex has even degree.
  A \emph{(simple) path} is a walk that contains each edge
  and
  internal vertex at most once.
  A \emph{cycle} is a closed path.
\end{definition}

\begin{proposition}\label[proposition]{prop:longcyc}
  Let $G$~be a connected \(n\)-vertex graph with positive integer edge
  weights and let $T$~be
  a closed walk of length \(2n-2+k\) containing
  all vertices of~\(G\).  Then, the weight of~$T$ is at least $W+k$,
  where \(W\)~is the minimum weight
  of any closed walk containing all vertices of~\(G\).
\end{proposition}

\noindent
To prove \cref{prop:longcyc}, we exploit the following theorem.
\begin{theorem}[{\citet[Satz~1]{Mad74}}]\label{mader}
  Every simple graph with minimum degree~\(k\geq 2\) contains a cycle~\(C\)
  such that there are \(k\)~mutually internally vertex-disjoint paths
  between any pair of vertices of~\(C\).
\end{theorem}

\begin{corollary}\label[corollary]{cor:shortcut}
  Let $G$~be a connected \(n\)-vertex multigraph without loops
  such that
  the deletion of the edges of any cycle disconnects~$G$.
  Then $G$~has at most $2n-2$~edges.
\end{corollary}

\begin{proof}
  We prove the statement by induction on~\(n\).
  The statement is trivial for $n=1$,
  since such a graph has no edges.
  Now, let \(n>1\).
  If \(G=(V,E)\)~contains a \emph{cutset}~\(E'\subseteq E\)
  of cardinality at most 2
  (that it, its deletion disconnects the graph), then
  $G\setminus E'=(V,E\setminus E')$~consists of
  two connected components~$G_1$ and~$G_2$.
  By induction, we have 
  $|E(G)|=|E(G_1)|+|E(G_2)|+|E'|\le 2|V(G_1)|-2+2|V(G_2)|-2+2=2n-2$.
  It remains to prove that $G=(V,E)$~indeed contains
  a cutset~$E'\subseteq E$ of cardinality at most 2.

  If \(G\)~contains a vertex~\(v\) of degree at most~2,
  then its incident edges form the sought cutset of cardinality~2.
  If $G$~contains at least two distinct edges~\(e_1\) and~\(e_2\)
  between the same pair~\(u\) and~\(v\) of vertices,
  then \((u,e_1,v,e_2,u)\)
  is a cycle and, thus, $E'=\{e_1,e_2\}$ is the sought cutset of cardinality~2.
  
  If none of the above apply,
  then \(G\)~is a simple graph with minimum degree three.
  Thus, by \cref{mader},
  \(G\)~contains a cycle~\(C\) whose edges can be deleted
  without disconnecting the graph:
  deleting the edges of~\(C\) removes
  at most two out of three pairwise internally vertex-disjoint paths
  between any pair of vertices of~\(C\).
  This contradicts our assumption that deleting the edges of any cycle disconnects~\(G\).
  \qed
\end{proof}

\noindent We can now prove \cref{prop:longcyc}.
 
\begin{proof}[of \cref{prop:longcyc}]
  Consider the following multigraph~$H$:
  the vertices of~\(H\) are the vertices of~\(G\)
  and
  the number of edges in~\(H\) between
  each pair~\(u\) and~\(v\) of vertices
  is equal to the number of times
  the closed walk~\(T\) contains the edge~$\{u,v\}$ of~\(G\).
  The multigraph~$H$ is connected, Eulerian, has \(n\)~vertices
  and $2n-2+k$ edges.
  If $k>0$, then, by \cref{cor:shortcut},
  there is a cycle~$C$ in~$H$ whose removal results
  in a connected multigraph~\(H'\).
  Multigraph~\(H'\) is still Eulerian.
  Thus,
  we can iterate the process until
  we get an Eulerian submultigraph~\(H^*\) of~\(H\) with at most $2n-2$~edges.
  The Euler tour~$T^*$ of~\(H^*\)
  is a closed walk for~\(G\) containing all its vertices
  and thus
  has weight at least~\(W\).
  Since the total weight of the deleted cycles is at least~$k$,
  the weight of~$T$ is at least $W+k$.
 \qed
\end{proof}
Finally, note that the bound given by \cref{prop:longcyc}
is tight for each~\(n\) and even~\(k\):
consider a unit-weight path graph~\(G\) on \(n\)~vertices \(v_1,\dots,v_n\).
The minimum weight of any closed walk containing all vertices is~\(W:=2n-2\).
The closed walk visiting the vertices
\(v_1,v_2,\dots,v_{n-1},v_n,v_{n-1},\dots,v_2,v_1,v_2,\dots,v_{k'-1},v_{k'},v_{k'-1},\dots,\allowbreak v_2,v_1\)
with \(k':=(k+2)/2\) in this order
has length~\(2n-2+2k'-2=2n-2+k\) and its weight is~\(W+2k'-2=W+k\).

\end{document}